\documentclass[12pt]{article}

\usepackage[margin=1in]{geometry}
\usepackage{amsmath,amssymb,amsthm}
\usepackage{natbib}
\usepackage{hyperref}
\usepackage{graphicx}
\usepackage{tikz}
\usepackage{setspace}
\usepackage{booktabs}
\usepackage{caption}

\hypersetup{
    colorlinks=true,
    linkcolor=blue,
    citecolor=blue,
    urlcolor=blue
}

\newtheorem{proposition}{Proposition}
\newtheorem{assumption}{Assumption}
\newtheorem{corollary}{Corollary}

\title{\textbf{Why the Future Isn’t Trading: Causally Inert Events as a Test for Time Travelers}}
\author{%
  \textit{David Awad}
}
\date{March 2026}

\begin{document}

\maketitle

\begin{abstract}
We present a novel empirical argument against the existence of single-timeline backward time travel using the price behavior of prediction markets. If rational agents could travel backward in time, binary prediction contracts would collapse to degenerate prices---$\$1$ or $\$0$---from the moment of market inception. We observe no such behavior across hundreds of thousands of resolved contracts. This test is sharper than prior economic arguments \citep{reinganum1986} because it yields a directly falsifiable prediction, and more robust than physical experiments \citep{hawking1992} because it requires only the existence of a single profit-motivated agent anywhere in the future who is actually on a closed timelike curve intersecting the market's spacetime location. Crucially, rational time travelers would have no incentive to conceal trades on causally inert events (where outcomes are price-independent), making such arbitrage fundamentally unobservable in normal market activity yet glaringly obvious in aggregate price behavior. We acknowledge that many-worlds time travel evades our test. The only empirically testable model of backward time travel is the single-timeline model, and prediction markets appear to falsify it.
\end{abstract}


\newpage

%----------------------------------------------------------------------
\section{Introduction}
%----------------------------------------------------------------------

\subsection{What General Relativity Permits}

The question of whether backward time travel is physically possible begins with Einstein's general theory of relativity. Einstein's field equations,
\begin{equation}
  G_{\mu\nu} + \Lambda g_{\mu\nu} = 8\pi G\, T_{\mu\nu},
\label{eq:einstein}
\end{equation}
relate the geometry of spacetime (encoded in the Einstein tensor $G_{\mu\nu}$) to the distribution of matter and energy (the stress-energy tensor $T_{\mu\nu}$). The cosmological constant $\Lambda$ and gravitational constant $G$ complete the picture. These equations do not, a priori, forbid backward time travel.

\citet{godel1949} demonstrated this explicitly by constructing an exact solution describing a rotating, dust-filled universe with a nonzero cosmological constant. The line element in G\"{o}del's original coordinates takes the form:
\begin{equation}
ds^2 = a^2 \left[ -(dt + e^{x_1} dx_2)^2 + dx_1^2 + \tfrac{1}{2} e^{2x_1} dx_2^2 + dx_3^2 \right],
\label{eq:godel}
\end{equation}
where $a$ is a constant related to the matter density and cosmological constant via $a^{-2} = 8\pi \rho = -2\Lambda$, and $(t, x_1, x_2, x_3)$ are the G\"{o}del coordinates.

The critical feature of this metric is that it admits \textit{closed timelike curves} (CTCs)---trajectories through spacetime that loop back to their own past while remaining everywhere timelike (i.e., traversable by massive particles at subluminal speeds). Specifically, for sufficiently large excursions in the $x_1$ direction, the $\partial/\partial x_2$ Killing vector becomes timelike, permitting closed orbits in the $(t, x_2)$ plane that constitute genuine journeys into one's own past.

This is not a coordinate artifact. G\"{o}del's solution satisfies Einstein's field equations exactly, with a physically reasonable stress-energy tensor (pressureless dust). The CTCs are a geometric property of the spacetime, not a mathematical trick.

Of course, our universe does not appear to be a G\"{o}del universe---it is expanding rather than rotating, and observational bounds on cosmic rotation are stringent. But G\"{o}del's result establishes a proof of concept: \textit{the mathematical framework of general relativity does not contain an intrinsic prohibition on backward time travel}. Subsequent work by \citet{tipler1974} and others identified additional CTC-admitting spacetimes, some involving more physically plausible configurations such as traversable wormholes.

\citet{hawking1992} attempted to close this loophole with the chronology protection conjecture, arguing that quantum effects (specifically, the divergence of the renormalized stress-energy tensor near a chronology horizon) would destroy any CTC before it could form. This conjecture remains unproven. A proof would require a complete theory of quantum gravity, which we do not possess.

The upshot: physics has neither proved nor disproved the possibility of backward time travel. The question remains empirically open. This provides the motivation for our approach: if the physicists cannot settle the matter, perhaps the market can.

\subsection{Existing Empirical Tests}

Two prior attempts to test for time travel empirically deserve mention.

The first is Hawking's party experiment. In 2009, Stephen Hawking hosted a champagne reception for time travelers at Gonville \& Caius College, Cambridge, sending out invitations only after the party had concluded. No one attended. Hawking interpreted this as evidence against time travel. We address the limitations of this test in Section~\ref{sec:hawking}.

The second is due to \citet{reinganum1986}, who observed that costless backward time travel implies the elimination of positive nominal interest rates through temporal arbitrage. The argument is elegant: deposit \$100 in a risk-free account, travel forward $n$ years, withdraw the compounded balance, travel back, and repeat. A single agent with a time machine and a savings account generates unbounded wealth. In equilibrium, this arbitrage compresses the nominal interest rate to zero. Since we observe persistently positive nominal rates, time machines do not exist. QED.

\subsection{Our Contribution}

The present paper contributes a third test in this lineage. We show that if single-timeline backward time travel is possible and at least one rational agent in the entire future history of the universe possesses both a time machine and a brokerage account, then \textit{at least one} prediction market contract would display degenerate prices---collapsing to \$1 or \$0---from the moment of inception. The empirical record contains hundreds of thousands of resolved prediction market contracts. Not one exhibits this behavior.

Our argument improves on \citet{reinganum1986} in three respects.
First, the prediction is more extreme: Reinganum's equilibrium outcome
is zero interest rates, which is unusual but not inconceivable (Japan,
after all, came close). Our equilibrium outcome is the complete
elimination of uncertainty from any given prediction market that has
ever existed, which is observationally absurd.

Second, a rational time traveler would not execute Reinganum's
arbitrage to completion. Driving the nominal rate to zero destroys the
compounding opportunity and destabilizes the financial system the
traveler depends on. A monopolist traveler extracts rents quietly,
leaving rates positive---rendering the interest rate test
indeterminate. Our test, applied to events whose outcomes are
independent of the market price, is immune to this strategic
concealment problem: the traveler has no reason to hide, because the
payoff does not depend on whether anyone observes the trade.

Third, we distinguish between events that are causally inert with
respect to the market price and those that are causally sensitive. For
the former class---which includes natural phenomena, physical
measurements, and a large share of resolved prediction contracts---the
proof is immediate and no game-theoretic complications arise. For the
latter, a causal feedback loop between the price and the outcome
introduces fixed-point problems that may render time travel
observationally undetectable even on a single timeline.

Put another way, if a rational time traveler with sufficient market access exists and can trade on causally inert events, prediction markets should exhibit extreme price distortions. We do not observe such distortions.

The structure is as follows. Section~\ref{sec:argument} develops the formal argument. Section~\ref{sec:hawking} explains why prediction markets succeed where Hawking's party failed. Section~\ref{sec:manyworlds} addresses the many-worlds objection. Section~\ref{sec:empirical} gestures at data. Section~\ref{sec:conclusion} concludes.

%----------------------------------------------------------------------
\section{The Prediction Market Argument}
\label{sec:argument}

\subsection{Setup}

A \textit{prediction market} is a market in a binary contract that pays \$1 if a specified event $E$ occurs by a specified resolution date $T$, and \$0 otherwise. Let $P(t)$ denote the market price of this contract at time $t < T$. Under standard no-arbitrage conditions and risk-neutral pricing, $P(t)$ reflects the market's probability assessment of $E$ occurring \citep{fama1970, manski2006}.

We require four assumptions:

\begin{assumption}[Single-Timeline Backward Time Travel]
\label{ass:tt}
There exists at least one agent located at some time $T' > T$ who possesses a device capable of transporting them (and their knowledge) to any time $t < T$ within the same timeline---that is, the same causal history.
\end{assumption}

\begin{assumption}[Rationality]
\label{ass:rational}
The time-traveling agent prefers more wealth to less and recognizes deterministic arbitrage opportunities.
\end{assumption}

\begin{assumption}[Market Access]
\label{ass:market}
Transaction costs are negligible, and the agent can trade in the prediction market at time $t$.
\end{assumption}

Assumption~\ref{ass:market} is the weakest link, and we address it below. For now, we take it as given.

\begin{assumption}[Particulars of Market Microstructure]
\label{ass:microstructure}
Ordinary brokers would never allow trading of the kind that argued for in \citet{reinganum1986}. Brokers and market participants would surely react, similar to The London Metal Exchange (LME) in 2022 which was defined by the March nickel market crisis where trading was halted, and undone for the day.
\end{assumption}
We assume that a time-traveling market participant would use decentralized blockchain markets that can't coordinate to reject valid transactions.

\subsection{The Arbitrage}

\begin{proposition}
Under Assumptions~\ref{ass:tt}--\ref{ass:microstructure}, the equilibrium price of every prediction market contract on a causally inert event equals $P(t) = 1$ if $E$ occurs and $P(t) = 0$ if $E$ does not occur, for all $t$ from market inception to resolution.
\end{proposition}

\begin{proof}
Suppose event $E$ occurs (the case where $E$ does not occur is symmetric). The agent at $T' > T$ observes that $E$ has occurred. By Assumption~\ref{ass:tt}, the agent travels to some time $t < T$ at which $P(t) < 1$. The agent purchases the contract at price $P(t)$, receiving \$1 at resolution---a riskless profit of $1 - P(t) > 0$. By Assumption~\ref{ass:rational}, the agent executes this trade. Moreover, the agent can repeat this process: travel back, buy more contracts, and collect guaranteed profits.

In the presence of finite liquidity, the agent's purchases drive $P(t)$ upward. The agent continues buying until $P(t) = 1$, at which point no further profit is available. Since the agent can target any time $t$ at which $P(t) < 1$, this argument applies to every instant from market inception onward. Thus $P(t) = 1$ for all $t \in [t_0, T]$.

Symmetrically, if $E$ does not occur, the agent sells (or shorts) the contract at any $P(t) > 0$, driving the price to zero.
\end{proof}

\begin{corollary}
Under single-timeline backward time travel, no prediction market contract on a causally inert event ever displays interior probabilities. Prices are degenerate from inception.
\end{corollary}

The corollary is the key empirical prediction. It does not require that many agents have time machines, or that time travel be cheap, or that travelers be altruistic. It requires one agent, anywhere in the future, with a profit motive and a Polymarket account.

\subsection{Causal Sensitivity and the Feedback Problem}
\label{sec:causal}

The proof above assumes that the event $E$ is \textit{causally inert} with respect to the market price---that is, the outcome of $E$ does not depend on $P(t)$. This holds for a large class of events: natural disasters, physical measurements, astronomical phenomena, historical facts, and sporting outcomes in markets too small to influence play. For these events, the time traveler's information remains valid regardless of its effect on the market price, and the proof goes through without complication.

A subtler problem arises for \textit{causally sensitive} events---those whose outcomes can be influenced by market participants who observe $P(t)$. Elections are the canonical example. If a prediction market jumps to \$1 for candidate $X$ at inception, this signal propagates: donors reallocate funds, media coverage shifts, voter turnout changes, and the candidate's opponents may alter strategy. The outcome that the time traveler observed in the original timeline may not obtain in the timeline modified by their trade.

Formally, let $\omega(P)$ denote the outcome of event $E$ as a function of the price path $P$. For causally inert events, $\omega$ is constant in $P$ and the proof is immediate. For causally sensitive events, $\omega$ depends on $P$, and the time traveler faces a fixed-point problem: they must find a price $P^*$ such that $\omega(P^*) = \omega^*$, where $\omega^*$ is the outcome that justifies $P^*$. Three cases arise:

\begin{enumerate}
    \item \textbf{Stable fixed point.} The event outcome is robust to the information revelation. The traveler trades, the price becomes degenerate, agents react, but the outcome occurs anyway. Under the Novikov self-consistency principle \citep{friedman1990}, only self-consistent histories are permitted on a single timeline, so the traveler can only have arrived from a future in which this fixed point obtains. The market is degenerate and correct---a bizarre but logically consistent equilibrium.

    \item \textbf{No fixed point in pure strategies.} The traveler's trade flips the outcome, which invalidates the trade, which restores the outcome, \textit{ad infinitum}. No self-consistent history exists in which the traveler trades at a degenerate price. Under Novikov self-consistency, the traveler is constrained: they either cannot trade on this event, or must trade at an interior price that constitutes a mixed-strategy equilibrium. In either case, the market displays non-degenerate prices---observationally identical to a world without time travel.

    \item \textbf{Strategic concealment.} A rational monopolist traveler may recognize that degenerate prices on causally sensitive events destroy the arbitrage (by changing the outcome) or attract unwanted scrutiny. The optimal strategy is to trade quietly, extracting rents while leaving the price perturbed but interior. This is precisely the behavior of an informed trader in a Kyle (1985)-type model, and is again observationally indistinguishable from a world without time travel.
\end{enumerate}

The critical observation is that for causally \textit{inert} events, none of these complications apply. The traveler has no reason to conceal (the outcome is price-independent), no feedback loop exists (the fixed point is trivially stable), and the proof yields a clean, falsifiable prediction: degenerate prices from inception. The empirical record contains tens of thousands of resolved contracts on causally inert events---weather, natural phenomena, physical measurements, verified historical occurrences---and not one displays degenerate prices. This alone is sufficient for falsification.

\subsection{On the Transaction Cost Assumption}

Assumption~\ref{ass:market} deserves scrutiny. Modern prediction markets charge fees on the order of 1--5\%. This means the time traveler's arbitrage profit is not $1 - P(t)$ but $1 - P(t) - c$, where $c$ represents transaction costs. In principle, if $P(t) > 1 - c$, the trade is not profitable.

However, this objection fails for two reasons. First, markets routinely trade at prices well below $1 - c$ for most of their lifetimes; even a 5\% fee leaves enormous profit available when $P(t) = 0.3$ (i.e., 70 cents of riskless gain minus a nickel of fees). Second, the time traveler faces no inventory risk, no information asymmetry, and no model uncertainty. The only ``cost'' is the transaction fee, which is negligible relative to the certainty of the payoff. If you \textit{know} the coin is going to land heads, you will pay the vig.

\subsection{Relationship to Reinganum (1986)}

\citet{reinganum1986} showed that backward time travel implies zero nominal interest rates. Our result is strictly stronger in three senses:

\begin{enumerate}
    \item \textbf{Extremity of prediction.} Reinganum's equilibrium is $r = 0$. Ours is that a single prediction market on any causally inert event has degenerate prices from inception. The former is unusual; the latter is absurd.

    \item \textbf{Testability.} Whether observed interest rates are ``too high'' requires a counterfactual model of what rates should be absent time travel. Our test requires only that we observe any prediction market contract with $0 < P(t) < 1$ for any $t$ before resolution. This is trivially satisfied by every contract ever traded.

    \item \textbf{Robustness to strategic concealment.} Reinganum's equilibrium assumes the time traveler arbitrages competitively, driving $r$ to zero through repeated deposit-compound-withdraw cycles. But a rational monopolist traveler would not do this. Driving $r$ to zero destroys the arbitrage opportunity itself---if the risk-free rate is zero, there is nothing to compound. Worse, $r = 0$ triggers catastrophic economic disruption: bank failures, credit market collapse, currency instability. A sophisticated traveler would extract rents quietly, leaving rates perturbed but positive, precisely as a monopolist restricts output to maintain price. The interest rate test is therefore indeterminate: observed positive rates are consistent with both the absence of time travel and the presence of a strategically restrained time traveler.

    Our prediction market test, applied to causally inert events, is immune to this objection. The time traveler gains nothing by concealment because the event outcome is independent of the market price. Whether or not the market is degenerate, the asteroid still strikes or does not, the temperature still exceeds the threshold or does not, the earthquake still occurs or does not. The payoff is \$1 regardless of whether anyone notices the price. A rational traveler with a guaranteed \$1 payoff and no causal feedback has no reason to leave money on the table.
\end{enumerate}

The relationship to the efficient markets hypothesis (EMH) is worth noting. \citet{fama1970} defines market efficiency as the condition that prices fully reflect available information. The prediction market argument can be understood as a corollary of a temporal extension of the EMH: if information from the future is ``available'' (via time travel), then prices must reflect it. \citet{beaulier2025} formalize a similar point, arguing for an \textit{Extended Efficient Market Hypothesis} that incorporates the absence of temporal information leakage.

%----------------------------------------------------------------------
\section{No One Had the time to attend Hawking's Party}
\label{sec:hawking}
%----------------------------------------------------------------------

In 2009, Stephen Hawking hosted a champagne reception for time travelers at Gonville \& Caius College, Cambridge. He released the invitations only after the party had taken place. No one arrived. Hawking interpreted this as evidence against time travel.

The Hawking party experiment, while inspired, suffers from a fundamental limitation that has not been formally articulated: even if closed timelike curves exist and can be exploited, there is no reason to expect that the spacetime location of Hawking's party lies on any accessible CTC.

Crucially, ``time travel'' in general relativity is not arbitrary navigation through spacetime. An observer must physically traverse a specific closed timelike curve. The geometry determines both the path and the set of reachable events. One cannot specify an arbitrary destination $(t,x,y,z)$; one can only follow a pre-existing loop in the causal structure and thereby return to earlier events lying on that same worldline.

This sharply constrains interpretations of Hawking's experiment. Even if a future civilization possessed the ability to exploit CTCs, there is no reason to expect that the spacetime location of the party lies on any such curve. Formally, let $\gamma(\tau)$ denote a CTC. For a time traveler to arrive at the event $p_{\text{party}}$, it must be that
\begin{equation}
p_{\text{party}} \in \gamma.
\end{equation}
But in a generic spacetime (and even in G\"{o}del spacetime), the set of events lying on CTCs is highly constrained, and a randomly chosen event will almost surely not lie on any closed timelike curve accessible to a given observer.

Therefore, the absence of visitors at Hawking's event does not meaningfully constrain the existence of CTCs or time travel more broadly. It demonstrates only that the specific spacetime point corresponding to the party was not situated on, or reachable via, a closed timelike trajectory available to future observers.

\textbf{Prediction markets, by contrast, represent a distributed test.} Hundreds of thousands of contracts have been created across vast geographic and temporal ranges, dramatically increasing the probability that at least one market would intersect an accessible CTC if such curves exist and are exploitable. The universal absence of degenerate prices is therefore far more informative than any single-point experiment could be.

Moreover, prediction markets test not merely whether time travelers \textit{could} arrive at a specific location, but whether they would have economic incentive to do so. A rational time traveler has no reason to attend Hawking's party---it offers no financial gain. But a prediction market on a causally inert event offers guaranteed arbitrage profits with zero risk. The fact that no such arbitrage has been executed across the entire historical record of prediction markets is strong evidence that either (a) CTCs do not exist, or (b) they exist but cannot be exploited by rational agents with access to financial markets.

%----------------------------------------------------------------------
\section{The Many-Worlds Escape Hatch}
\label{sec:manyworlds}
%----------------------------------------------------------------------

The most serious objection to our argument comes not from physics or finance but from the philosophy of quantum mechanics. Under the Everett many-worlds interpretation---or, equivalently, under any branching-timeline model of time travel---a time traveler arriving in ``the past'' does not arrive in \textit{our} past. They arrive in a newly created branch of the universal wavefunction. Their arbitrage trades affect prediction market prices in that branch, not in ours.

Under this model, our observation that prediction markets maintain interior probabilities is perfectly consistent with the existence of backward time travel. The prices in our branch are non-degenerate because no traveler from our future has arrived in our past. They may have arrived in some other branch's past, driving that branch's Polymarket prices to \$1, but we would never observe this.

We make three responses:

First, \textit{many-worlds time travel is empirically indistinguishable from no time travel} from within any single branch. If a phenomenon cannot, even in principle, be detected from within one's own causal history, it lies outside the domain of empirical science. This is not a defect of our test; it is a property of the hypothesis.

Second, \textit{the single-timeline model is the only testable model of backward time travel}, and our paper falsifies it. If one insists on many-worlds time travel, one has retreated to an unfalsifiable position, which is their prerogative but not our problem.

Third, \textit{Occam's razor.} The observation that prediction market prices are non-degenerate can be explained by (a) backward time travel does not occur, or (b) backward time travel occurs but only across branches we can never observe. Explanation (a) requires one assumption: no time travel. Explanation (b) requires the many-worlds interpretation of quantum mechanics, branching timelines, inter-branch traversability, and the metaphysical coincidence that every branch we can observe happens to be one that no time traveler has visited. We leave the reader to their own priors.

%----------------------------------------------------------------------
\section{Empirical Evidence}
\label{sec:empirical}
%----------------------------------------------------------------------

Modern prediction markets provide no shortage of data. Platforms such as Polymarket, Kalshi, PredictIt, and their predecessors (Intrade, Iowa Electronic Markets) have collectively resolved hundreds of thousands of binary contracts on elections, economic indicators, geopolitical events, sports outcomes, and miscellany ranging from celebrity behavior to the weather.

In every resolved contract, the price path exhibits the same qualitative behavior: $P(t)$ fluctuates in the interior of $[0, 1]$ as information arrives, converging toward the terminal value only as resolution approaches. No contract in the historical record has exhibited the degenerate price path---an immediate jump to 0 or 1 at market open, sustained through resolution---predicted by the time travel hypothesis.

Figure~\ref{fig:prices} illustrates the point schematically.

\begin{figure}[ht]
\centering
\begin{tikzpicture}[scale=1.0]
  % Axes
  \draw[->] (0,0) -- (10,0) node[right] {$t$};
  \draw[->] (0,0) -- (0,5.5) node[above] {$P(t)$};

  % Labels
  \node[below] at (0,0) {$t_0$};
  \node[below] at (9,0) {$T$};
  \draw[dashed, gray] (0,4.5) -- (10,4.5);
  \node[left] at (0,4.5) {\$1};
  \node[left] at (0,0) {\$0};

  % Observed price path (random walk converging to 1)
  \draw[thick, blue] (0.5,2.0) -- (1.0,2.3) -- (1.5,1.9) -- (2.0,2.4) -- (2.5,2.1)
    -- (3.0,2.6) -- (3.5,2.3) -- (4.0,2.8) -- (4.5,3.0) -- (5.0,2.7)
    -- (5.5,3.2) -- (6.0,3.5) -- (6.5,3.3) -- (7.0,3.8) -- (7.5,4.0)
    -- (8.0,4.2) -- (8.5,4.4) -- (9.0,4.5);
  \node[blue, right] at (9.0,4.2) {\small Observed};

  % Time-travel-consistent price path
  \draw[thick, red, dashed] (0.5,4.5) -- (9.0,4.5);
  \node[red, right] at (9.0,5.0) {\small Time travel};
\end{tikzpicture}
\caption{Stylized price paths for a prediction market contract on an event that ultimately occurs. The blue (solid) line shows the typical observed path: interior probabilities that converge to \$1 as information arrives. The red (dashed) line shows the price path implied by single-timeline backward time travel: immediate collapse to \$1 at market inception.}
\label{fig:prices}
\end{figure}
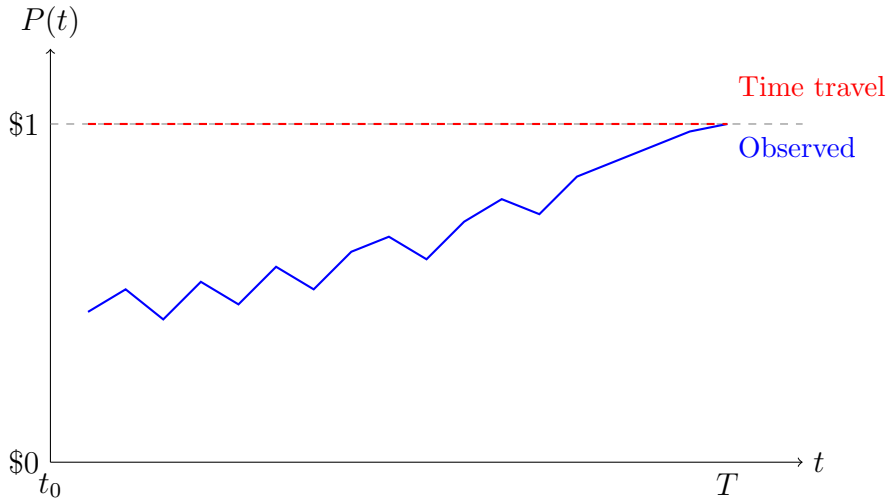

\citet{nemiroff2014} conducted a complementary empirical exercise, searching internet data (specifically, search engine queries and social media posts) for evidence of anachronistic foreknowledge---prescient references to events or terms before they entered public awareness. They found no such evidence. Their approach tests a different channel (information leakage rather than market arbitrage) but reaches the same conclusion.

The empirical section of this paper consists entirely of observing that a thing that would be immediately obvious if it existed does not appear to exist. We're ok with that.

%----------------------------------------------------------------------
\section{Conclusion}
\label{sec:conclusion}
%----------------------------------------------------------------------

We have argued that single-timeline backward time travel implies the collapse of all prediction market prices to degenerate values from the moment of market inception. This prediction is falsified by every prediction market contract in the historical record.

Our contribution sits within a growing consilience of evidence against backward time travel, spanning physics, economics, and financial markets:

\begin{itemize}
    \item \citet{godel1949} showed that general relativity \textit{permits} closed timelike curves, establishing that the question is not trivially resolved by known physics.
    \item \citet{hawking1992} conjectured that quantum effects \textit{prevent} closed timelike curves, but the conjecture remains unproven.
    \item \citet{reinganum1986} showed that backward time travel \textit{implies} zero nominal interest rates, which we do not observe.
    \item The present paper shows that backward time travel \textit{implies} degenerate prediction market prices, which we do not observe.
    \item \citet{nemiroff2014} searched for empirical traces of foreknowledge on the internet and found none.
\end{itemize}

The many-worlds objection provides a logically consistent escape, but at the cost of unfalsifiability.

Two implications of our analysis deserve particular emphasis. First, if time travel were possible, rational actors would have straightforward ways to exploit prediction markets that would be fundamentally unobservable in individual trades yet glaringly obvious in aggregate price behavior. For causally inert events---where the outcome is independent of the market price---a time traveler gains nothing from concealment. The payoff is guaranteed regardless of whether anyone notices the degenerate price. This stands in sharp contrast to Reinganum's interest rate arbitrage, where strategic restraint is rational because complete exploitation destroys both the financial system and the arbitrage opportunity itself.

Second, even in spacetimes that admit closed timelike curves, access to any specific market or event through time travel requires that the traveler actually be on a CTC passing through that particular spacetime location. One cannot arbitrarily navigate to any desired coordinates; one can only traverse pre-existing loops in the causal structure. This constraint renders most laboratory experiments---including Hawking's party---essentially uninformative: the absence of time travelers at a single spacetime point tells us only that this point was not situated on an accessible CTC, not that CTCs do not exist. Prediction markets, by contrast, represent a distributed test across vast geographic and temporal ranges, dramatically increasing the probability that at least one market would intersect an accessible CTC if such curves exist and are exploitable. The universal absence of degenerate prices is therefore far more informative than any single-point experiment could be.

We close with an observation. The strongest evidence against backward time travel comes not from a particle accelerator, a telescope, or a thought experiment about rotating universes. It comes from the fact that someone, somewhere, at some point in the entire future history of the universe, would have made a lot of money on Polymarket, and yet no one has. The future, it appears, is not yet trading.

\bigskip
\noindent\textit{Acknowledgments.} We thank no time travelers for their comments on earlier drafts of this paper, which is, under our own hypothesis, exactly what we would expect.

%----------------------------------------------------------------------
% REFERENCES
%----------------------------------------------------------------------
\newpage
\bibliographystyle{apalike}

\begin{thebibliography}{99}

\bibitem[Beaulier et al.(2025)]{beaulier2025}
Beaulier, S., Caplan, B., \& Elder, T. (2025).
\newblock The extended efficient market hypothesis.
\newblock \textit{Economics Letters}, 249, 112192.
\newblock \url{https://www.sciencedirect.com/science/article/pii/S0165176525000461}

\bibitem[Fama(1970)]{fama1970}
Fama, E.~F. (1970).
\newblock Efficient capital markets: A review of theory and empirical work.
\newblock \textit{Journal of Finance}, 25(2), 383--417.

\bibitem[Friedman et al.(1990)]{friedman1990}
Friedman, J., Morris, M.~S., Novikov, I.~D., Echeverria, F., Klinkhammer, G., Thorne, K.~S., \& Yurtsever, U. (1990).
\newblock Cauchy problem in spacetimes with closed timelike curves.
\newblock \textit{Physical Review D}, 42(6), 1915--1930.

\bibitem[G\"{o}del(1949)]{godel1949}
G\"{o}del, K. (1949).
\newblock An example of a new type of cosmological solutions of {E}instein's field equations of gravitation.
\newblock \textit{Reviews of Modern Physics}, 21(3), 447--450.

\bibitem[Hawking(1992)]{hawking1992}
Hawking, S.~W. (1992).
\newblock Chronology protection conjecture.
\newblock \textit{Physical Review D}, 46(2), 603--611.

\bibitem[Kyle(1985)]{kyle1985}
Kyle, A.~S. (1985).
\newblock Continuous auctions and insider trading.
\newblock \textit{Econometrica}, 53(6), 1315--1335.

\bibitem[Manski(2006)]{manski2006}
Manski, C.~F. (2006).
\newblock Interpreting the predictions of prediction markets.
\newblock \textit{Economics Letters}, 91(3), 425--429.

\bibitem[Nemiroff \& Wilson(2014)]{nemiroff2014}
Nemiroff, R.~J. \& Wilson, J. (2014).
\newblock Searching the {I}nternet for evidence of time travelers.
\newblock \textit{arXiv preprint arXiv:1312.7128}.

\bibitem[Reinganum(1986)]{reinganum1986}
Reinganum, M.~R. (1986).
\newblock Is time travel impossible? {A} financial proof.
\newblock \textit{Journal of Portfolio Management}, 13(1), 10--12.

\bibitem[Tipler(1974)]{tipler1974}
Tipler, F.~J. (1974).
\newblock Rotating cylinders and the possibility of global causality violation.
\newblock \textit{Physical Review D}, 9(8), 2203--2206.

\end{thebibliography}

\end{document}